\documentclass[fullpage,12]{article}
\usepackage{ifthen}
\usepackage{graphicx}
\usepackage{amsmath}
\usepackage{amsthm}
\usepackage{amssymb}
\usepackage{mathtools}
\usepackage{paralist}
\usepackage{bm}
\usepackage{xspace}
\usepackage{url}
\usepackage{fullpage, prettyref}
\usepackage{boxedminipage}
\usepackage{wrapfig}
\usepackage{color}
\usepackage{xcolor}
\usepackage{framed}
\usepackage[pagebackref,letterpaper=true,colorlinks=true,pdfpagemode=none,urlcolor=blue,linkcolor=blue,citecolor=violet,pdfstartview=FitH]{hyperref}
\usepackage{cleveref}
\usepackage{enumitem}
\usepackage{times}

\usepackage{thmtools}
\usepackage{thm-restate}
\newtheorem{theorem}{Theorem}

\newtheorem{lemma}{Lemma}[section]

\newtheorem{corollary}[theorem]{Corollary}

\newtheorem{remark}{Remark}[section]

\newcommand{\ignore}[1]{}

\newcommand{\cA}{{\cal A}}
\newcommand{\cB}{\mathcal{B}}

\newcommand{\Z}{{\mathbb Z}}
\newcommand{\eps}{\varepsilon}

\newcommand{\bD}{\mathbf{D}}

\newcommand{\br}{\mathbf{r}}
\newcommand{\bc}{\mathbf{c}}
\newcommand{\bone}{\mathbf{1}}
\newcommand{\piR}{\pi_\br}
\newcommand{\piC}{\pi_\bc}

\newcommand{\boneR}{\piR}
\newcommand{\boneC}{\piC}

\newcommand{\ceil}[1]{\lceil#1\rceil}

\newcommand{\Sec}[1]{\hyperref[sec:#1]{\S\ref*{sec:#1}}} 
\newcommand{\Eqn}[1]{\hyperref[eq:#1]{(\ref*{eq:#1})}} 
\newcommand{\Fig}[1]{\hyperref[fig:#1]{Fig.\,\ref*{fig:#1}}} 
\newcommand{\Tab}[1]{\hyperref[tab:#1]{Tab.\,\ref*{tab:#1}}} 
\newcommand{\Thm}[1]{\hyperref[thm:#1]{Theorem\,\ref*{thm:#1}}} 
\newcommand{\Fact}[1]{\hyperref[fact:#1]{Fact\,\ref*{fact:#1}}} 
\newcommand{\Lem}[1]{\hyperref[lem:#1]{Lemma\,\ref*{lem:#1}}} 
\newcommand{\Prop}[1]{\hyperref[prop:#1]{Prop.~\ref*{prop:#1}}} 
\newcommand{\Cor}[1]{\hyperref[cor:#1]{Corollary~\ref*{cor:#1}}} 
\newcommand{\Conj}[1]{\hyperref[conj:#1]{Conjecture~\ref*{conj:#1}}} 
\newcommand{\Def}[1]{\hyperref[def:#1]{Definition~\ref*{def:#1}}} 
\newcommand{\Alg}[1]{\hyperref[alg:#1]{Alg.~\ref*{alg:#1}}} 
\newcommand{\Ex}[1]{\hyperref[ex:#1]{Ex.~\ref*{ex:#1}}} 
\newcommand{\Clm}[1]{\hyperref[clm:#1]{Claim~\ref*{clm:#1}}} 

\def\nnz{\mathsf{nnz}}

\renewcommand{\epsilon}{\eps}
\newcommand{\error}{\mathsf{error}}

\newcommand{\total}{h}
\colorlet{shadecolor}{blue!05}

\begin{document}
\title{Better and Simpler Error Analysis of the \\ Sinkhorn-Knopp Algorithm for Matrix Scaling}
\author{Deeparnab Chakrabarty \\ Dartmouth College \\ deeparnab@dartmouth.edu \and Sanjeev Khanna\thanks{This work was supported in part by the National Science Foundation grants CCF-1552909 and CCF-1617851.}  \\ University of Pennsylvania\\ sanjeev@cis.upenn.edu}
\maketitle
\begin{abstract}
	Given a non-negative $n \times m$ real matrix $A$, the {\em matrix scaling} problem is to determine if it is possible to scale the rows and columns so that each row and each column sums to a specified target value for it. 
	The matrix scaling problem arises in many algorithmic applications, perhaps most notably as a preconditioning step in solving linear system of equations. One of the most natural and by now classical approach to matrix scaling is the Sinkhorn-Knopp algorithm (also known as the RAS method) where one alternately scales either all rows or all columns to meet the target values. In addition to being extremely simple and natural, another appeal of this procedure is that it easily lends itself to parallelization. A central question is to understand the rate of convergence of the Sinkhorn-Knopp algorithm. 
	
	Specifically, given a suitable error metric to measure deviations from target values, and an error bound $\varepsilon$, how quickly does the Sinkhorn-Knopp algorithm converge to an error below $\varepsilon$? While there are several non-trivial convergence results known about the Sinkhorn-Knopp algorithm, perhaps somewhat surprisingly, even for natural error metrics such as $\ell_1$-error or $\ell_2$-error, this is not entirely understood. In this paper, we present 
	an elementary convergence analysis for the Sinkhorn-Knopp algorithm that improves upon the previous best bound. In a nutshell, our approach is to show (i) a simple bound on the number of iterations needed so that the KL-divergence between the current row-sums and the target row-sums drops below a specified threshold $\delta$, and (ii) then show that for a suitable choice of $\delta$, whenever KL-divergence is below $\delta$, then the $\ell_1$-error or the $\ell_2$-error is below $\varepsilon$. The well-known Pinsker's inequality immediately allows us to translate a bound on the KL divergence to a bound on $\ell_1$-error. To bound the $\ell_2$-error in terms of the KL-divergence, we establish a new inequality, referred to as \eqref{eq:pinskerp}. This inequality is a strengthening of Pinsker's inequality and may be of independent interest. 
	Our analysis of $\ell_2$-error significantly improves upon the best previous convergence bound for $\ell_2$-error. 
	
\end{abstract}
\section{Introduction}
In the matrix scaling problem one is given an $n\times m$ non-negative matrix $A$, and positive integer vectors $\br \in \Z^n_{>0}$ and $\bc\in \Z^m_{>0}$ with the same $\ell_1$ norm $\sum_{i=1}^n \br_i = \sum_{j=1}^m \bc_j = \total$.
The objective is to determine if there exist diagonal matrices $R\in \mathbb{R}^{n\times n}$ and $S\in \mathbb{R}^{m\times m}$ such that the $i$th row of the matrix $RAS$ sums to $\br_i$ for all $1\le i\le n$ {\em and} the $j$th column of $RAS$ sums to $\bc_j$ for all $1\le j\le m$.
Of special importance is the case when $n=m$ and $\br \equiv \bc \equiv \mathbf{1}_n$, the $n$-dimensional all-ones vector --
the $(\bone,\bone)$-matrix scaling problem wishes to scale the rows and columns of $A$ to make it doubly stochastic.
This problem arises in many different areas ranging from transportation planning~\cite{DemingStephan, OrtuzarWillumsen} to quantum mechanics~\cite{Schrodinger, Aaronson}; we refer the reader to a recent comprehensive survey by Idel~\cite{Idel} for more examples.

One of the most natural algorithms for the matrix scaling problem is the following Sinkhorn-Knopp algorithm~\cite{Sinkhorn,Sinkhorn-Knopp}, which is known by many names including the RAS method~\cite{Bachrach} and the Iterative Proportional Fitting Procedure~\cite{Ruschendorf1995}.
The algorithm starts off by multiplicatively scaling all the columns by the columns-sum times $\bc_j$ to get a matrix $A^{(0)}$ with column-sums $\bc$. Subsequently, for $t\geq 0$, it obtains the $B^{(t)}$ by scaling each row of $A^{(t)}$ by the respective row-sum times $\br_i$, and obtain $A^{(t+1)}$ by scaling each column of $B^{(t)}$ by the respective column sums time $\bc_j$. More precisely,
\[
A^{(0)}_{ij} := \frac{A_{ij}}{\sum_{i=1}^n A_{ij}}\cdot \bc_j ~~~~~~\forall t\geq 0,~~~ B^{(t)}_{ij} := \frac{A^{(t)}_{ij}}{\sum_{j=1}^m A^{(t)}_{ij}}\cdot \br_i, ~~~ A^{(t+1)}_{ij} := \frac{B^{(t)}_{ij}}{\sum_{i=1}^n B^{(t)}_{ij}}\cdot \bc_j
\]

The above algorithm is simple and easy to implement and each iteration takes $O(\nnz(A))$, the number of non-zero entries of $A$. 
Furthermore, it has been known for almost five decades~\cite{Sinkhorn, Sinkhorn-Knopp,FranklinLorenz,Soules} that if $A$ is $(\br,\bc)$-scalable then the above algorithm asymptotically\footnote{Computationally, this asymptotic viewpoint is unavoidable in the sense that there are simple examples for which the unique matrix scaling matrices need to have irrational entries.
	For instance, consider the following example from Rothblum and Schneider~\cite{RothSch}. The matrix is 
	$\begin{bmatrix}
	1 & 1 \\
	1 & 2
	\end{bmatrix}$ 
	with $\br\equiv\bc\equiv [1,1]^\top$. The unique $R$ and $S$ matrices are 
	$\begin{bmatrix}
	(\sqrt{2}+1)^{-1} & 0 \\
	0 & (\sqrt{2}+2)^{-1}
	\end{bmatrix}$ and
	$\begin{bmatrix}
	\sqrt{2} & 0 \\
	0 & 1
	\end{bmatrix}$, respectively, giving $RAS = 
	\begin{bmatrix}
	2 - \sqrt{2} & \sqrt{2} - 1 \\
	\sqrt{2} - 1 & 2 - \sqrt{2}
	\end{bmatrix}$.} converges to a right solution. More precisely, 
given $\eps > 0$, there is some finite $t$ by which 
one obtains a matrix which is ``$\eps$-close to having row- and column-sums $\br$ and $\bc$''. 

However, the rate of convergence of this simple algorithm is still not fully understood.
Since the rate depends on how we measure ``$\eps$-closeness'', we  look at two natural error definitions. 
For any $t$, let $\br^{(t)} := A^{(t)}\bone_m$ denote the vector of row-sums of $A^{(t)}$.
Similarly, we define $\bc^{(t)}:= {B^{(t)}}^\top\bone_n$ to be the vector of the column-sums of $B^{(t)}$.
Note that $\sum_{i=1}^n \br^{(t)}_i = \sum_{j=1}^m \bc^{(t)}_j = \total$ for all $t$.
The error of the matrix $A_t$ (the error of matrix $B_t$ similarly defined) is
\[
\ell_1\textrm{-error}:~~ \error_1(A_t) := ||\br^{(t)}-\br||_1 ~~~~~~~~~~~~~ \ell_2\textrm{-error}:~~\error_2(A_t) := ||\br^{(t)} - \br||_2
\]
\noindent
In this note, we give simple convergence analysis for both error norms. 
Our result is the following. 

\begin{theorem}\label{thm:1}
	Given a matrix $A\in \mathbb{R}^{n\times m}_{\geq 0}$ which is $(\br,\bc)$-scalable, and any $\epsilon> 0$, the Sinkhorn-Knopp algorithm 
	\begin{enumerate}
		\item in time $t = O\left(\frac{h^2\ln\left(\Delta\rho/\nu\right)}{\eps^2}\right)$ returns a matrix $A_t$ or $B_t$ with $\ell_1$-error $\leq \eps$.
		\item in time $t = O\left(\rho h \ln\left(\Delta\rho/\nu\right) \cdot \left(\frac{1}{\eps} + \frac{1}{\eps^2}\right)\right)$ returns a matrix $A_t$ or $B_t$ with $\ell_2$-error $\leq \eps$.
	\end{enumerate}
	Here $h = \sum_{i=1}^n \br_i = \sum_{j=1}^m \bc_j$, $\rho = \max(\max_i \br_i, \max_j \bc_j)$, $\nu = \frac{\min_{i,j:A_{ij}>0} A_{ij}}{\max_{i,j} A_{ij}}$, and $\Delta = \max_j |\{i: A_{ij} > 0\}|$ is the maximum number of non-zeros in any column of $A$.
\end{theorem}
\noindent
For the special case of $n = m$ and $\br\equiv\bc\equiv\bone_n$, we get the following as a corollary.
\begin{corollary}\label{cor:1}
	Given a matrix $A\in \mathbb{Z}^{n\times n}_{\geq 0}$ which is $(\bone,\bone)$-scalable, and any $\epsilon> 0$, the Sinkhorn-Knopp algorithm 
	\begin{enumerate}
		\item in time $t = O\left(\frac{n^2\ln(\Delta/\nu)}{\eps^2}\right)$ returns a matrix $A_t$ or $B_t$ with $\ell_1$-error $\leq \eps$.
		\item in time $t = O\left(n \ln(\Delta/\nu) \cdot \left(\frac{1}{\eps} + \frac{1}{\eps^2}\right)\right)$ returns a matrix $A_t$ or $B_t$ with $\ell_2$-error $\leq \eps$.
	\end{enumerate}
Here $\Delta = \max_j |\{i: A_{ij} > 0\}|$ is the maximum number of non-zeros in any column of $A$.
\end{corollary}
\begin{remark}\label{rem:prev-best}
	To our knowledge, the $\ell_1$-error hasn't been explicitly studied in the literature (but see last paragraph of Section~\ref{sec:persp}), although for small $\eps \in (0,1)$ the same can be
	deduced from previous papers on matrix scaling~\cite{LSW00,GurvitsY,KLRS08,KK93}. One of our main motivations to look at $\ell_1$-error arose from the connections to perfect matchings in bipartite graphs as observed by Linial, Samorodnitsky and Wigderson~\cite{LSW00}.
	For the $\ell_2$ error, which is the better studied notion in the matrix scaling literature, the best analysis is due to Kalantari et al~\cite{KLRS,KLRS08}. They give a $\tilde{O}(\rho h^2/\eps^2)$ upper bound on the number of iterations for the general problem, and for the special case when $m=n$ and the square matrix has positive permanent (see ~\cite{KLRS}), they give a $\tilde{O}(\rho (h^2 - nh + n)/\eps^2)$ upper bound. Thus, for $(\bone,\bone)$-scaling, they get the same result as in Corollary~\ref{cor:1}.
	We get a quadratic improvement on $h$ in the general case, and we think our proof is more explicit and simpler.
\end{remark}
\begin{remark}
	Both parts of Theorem~\ref{thm:1} and Corollary~\ref{cor:1} are interesting in certain regimes of error. When the error $\eps$ is ``small'' (say, $\leq 1$) so that $1/\eps^2 \geq 1/\eps$,
	then statement 2 of Corollary~\ref{cor:1} implies statement 1 by Cauchy-Schwarz. 
	However, this breaks down when $\eps$ is ``large'' (say $\eps = \delta n$ for some constant $\delta>0$). In that case, statement 1 implies that in $O(\ln n/\delta^2)$ iterations, the $\ell_1$-error is $\leq \delta n$, but Statement 2 only implies that in $O(\ln n/\delta^2)$ iterations, the $\ell_2$ norm is $\leq \delta n$. This ``large $\ell_1$-error regime'' is of particular interest for an application to approximate matchings in bipartite graphs discussed below.
\end{remark}

%

\paragraph{Applications to Parallel Algorithms for Bipartite Perfect Matching.}

As a corollary, we get the following application, first pointed by Linial et al~\cite{LSW00}, to the existence of perfect matchings in bipartite graphs. Let $A$ be the adjacency matrix of a bipartite graph $G = (L\cup R, E)$ with $A_{ij} = 1$ iff $(i,j)\in E$. If $G$ has a perfect matching, then clearly there is a doubly stochastic matrix $X$ in the support of $A$. This suggests the algorithm of running the Sinkhorn-Knopp algorithm to $A$, and the following claim suggests when to stop. Note that each iteration can be run in $O(1)$ parallel time with $m$-processors where $m$ is the number of edges.
\begin{lemma}
	If we find a column (or row) stochastic matrix $Y$ in the support of $A$ such that $\error_1(Y) \leq n\eps$, then $G$ has a matching of size $\geq n(1-\eps)$.
\end{lemma}
\begin{proof}
	Suppose $Y$ is column stochastic.
	Given $S\subseteq L$, consider $\sum_{i\in S,j\in \Gamma S} Y_{ij} = |S| + \sum_{i\in S} \left(\sum_{j=1}^n Y_{ij} - 1\right) \geq 
	|S| - \sum_{i=1}^n \Big|\sum_{j=1}^n Y_{ij} - 1\Big| \geq |S| - \error_1(Y) \geq |S| - n\eps$. On the other hand, $\sum_{i\in S,j\in \Gamma S} Y_{ij} \leq \sum_{j\in \Gamma S} \sum_{i=1}^n Y_{ij} = |\Gamma S|$. Therefore, for every $S\subseteq L$, $|\Gamma S| \geq |S| - n\eps$. The claim follows by approximate Hall's theorem.
\end{proof}
\noindent
%
\begin{corollary}[Fast Parallel Approximate Matchings]\label{thm:bip}
	Given a bipartite graph $G$ of max-degree $\Delta$ and an $\eps \in (0,1)$, $O(\ln \Delta/\eps^2)$-iterations of Sinkhorn-Knopp algorithm suffice to distinguish between the case when $G$ has a perfect matching and the case when the largest matching in $G$ has size at most $n(1-\eps)$.
\end{corollary}

Thus the approximate perfect matching problem in bipartite graphs is in NC for $\eps$ as small as polylogarithmic in $n$. 
This is not a new result and can indeed be obtained from the works on parallel algorithms for packing-covering LPs~\cite{LubyNisan, Young01, AllenZhuOrecchia15,RaoICALP16}, but the Sinkhorn-Knopp algorithm is arguably simpler.

\subsection{Perspective}\label{sec:persp}
As mentioned above, the matrix scaling problem and in particular the Sinkhorn-Knopp algorithm has been extensively studied over the past 50 years. We refer the reader to Idel's survey~\cite{Idel} and the references within for a broader perspective; in this subsection we mention the most relevant works. 

We have already discussed the previously best  known, in their dependence on $h$, analysis for the Sinkhorn-Knopp algorithm in Remark~\ref{rem:prev-best}. For the special case of {\em strictly positive} matrices, better rates are known.
Kalantari and Khachiyan~\cite{KK93}
showed that for positive matrices and the $(\bone,\bone)$-scaling problem, the Sinkhorn-Knopp algorithm obtains $\ell_2$ error $\leq \eps$ in $O(\sqrt{n}\ln(1/\nu)/\eps)$-iterations; this result was extended to the general matrix scaling problem by Kalantari et al~\cite{KLRS08}. In a different track, Franklin and Lorenz~\cite{FranklinLorenz} show that in fact the dependence on $\eps$ can be made 
logarithmic, and thus the algorithm has ``linear convergence'', however their analysis\footnote{~\cite{FranklinLorenz} never make the base of the logarithm explicit, but their proof shows it can be as large as $1 - 1/\nu^2$.} has a polynomial dependence of $(1/\nu)$.  All these results use the positivity crucially and seem to break down even with one $0$ entry.

%

The Sinkhorn-Knopp algorithm has polynomial dependence on the error parameter and therefore is a ``pseudopolynomial'' time approximation. 
We conclude by briefly describing bounds obtained by other algorithms for the matrix scaling problem whose dependence on $\eps$ is logarithmic rather than polynomial. Kalantari and Khachiyan~\cite{KalantariKhachiyan} describe a method based on the ellipsoid algorithm which runs in time $O(n^4 \ln(n/\eps)\ln(1/\nu))$. Nemirovskii and Rothblum~\cite{NemRoth} describe a method with running time $O(n^4 \ln(n/\eps)\ln\ln(1/\nu))$.
The first strongly polynomial time approximation scheme (with no dependence on $\nu$) was due to Linial, Samoridnitsky, and Wigderson~\cite{LSW00} who gave a $\tilde{O}(n^7\ln(h/\eps))$ time algorithm. 
Rote and Zachariasen~\cite{RoteZ} reduced the matrix scaling problem to flow problems to give a $O(n^4\ln(h/\eps))$ time algorithms for the matrix scaling problem. To compare, we should recall that Theorem~\ref{thm:1} shows that our algorithm runs in time $O(\nnz(A)h^2/\eps^2)$ time.

Very recently, two independent works obtain vastly improved running times for matrix scaling. Cohen et al~\cite{CohenFOCS} give $\tilde{O}(\nnz(A)^{3/2})$ time algorithm, while Allen-Zhu et al~\cite{AllenzhuFOCS} give a $\tilde{O}(n^{7/3} + \nnz(A)\cdot(n + n^{1/3}h^{1/2}))$ time algorithm; the tildes in both the above running times hide the logarithmic dependence on $\eps$ and $\nu$. Both these algorithms look at the matrix scaling problem as a convex optimization problem and perform second order methods. 
After the first version of this paper was made public, we were pointed out another recent paper by Altschuler, Weed and Rigollet~\cite{AWR-NIPS17}
who also study the $\ell_1$-error and obtain the same result as part 1 of our Theorem. Indeed their proof techniques are very similar to what we use to prove part 1.

\section{Entropy Minimization Viewpoint of the Sinkhorn-Knopp Algorithm}
There have been many approaches (see Idel~\cite{Idel}, Section 3 for a discussion) towards analyzing the Sinkhorn-Knopp algorithm including convex optimization and log-barrier methods~\cite{KK93, KLRS08, Macgill, BalakrishnanHT04}, non-linear Perron-Frobenius theory~\cite{Menon, Soules, FranklinLorenz, BrualdiPS, KK93}, topological methods~\cite{Raghavan, Bapat-Raghavan}, connections to the permanent~\cite{LSW00, KLRS}, and the entropy minimization method~\cite{Bregman, Csiszar1,Csiszar2,GurvitsY} which is what we use for our analysis.

We briefly describe the entropy minimization viewpoint. Given two non-negative matrices $M$ and $N$ let us define the {\em Kullback-Leibler} divergence\footnote{The KL-divergence is normally stated between two distributions and doesn't have the $1/h$ factor. Also the logarithms are usually base $2$.} between $M$ and $N$ as follows
\begin{equation}\label{eq:KL}
\bD(M,N) := \frac{1}{\total} \sum_{1\leq i\leq n}~\sum_{1\leq j\leq m} M_{ij} \ln\left(\frac{M_{ij}}{N_{ij}}\right)
\end{equation}
with the convention that the summand is zero if both $M_{ij}$ and $N_{ij}$ are $0$, and is $\infty$ if $M_{ij} > 0$ and $N_{ij} = 0$.
Let $\Phi_r$ be the set of $n\times m$ matrices whose row-sums are $\br$ and let $\Phi_c$ be the set of $n\times m$ matrices whose column sums are $\bc$. Given matrix $A$ suppose we wish to find the matrix $A^* = \arg \min_{B\in \Phi_r \cap \Phi_c} \bD(B,A)$. One algorithm for this is to use the method of alternate projections with respect to the KL-divergence~\cite{Bregman} (also known as $I$-projections~\cite{Csiszar1}) which alternately finds the matrices in $\Phi_r$ and $\Phi_c$ closest in the KL-divergence sense to the current matrix at hand, and then sets the minimizer to be the current matrix. It is not too hard to see (see Idel~\cite{Idel}, Observation 3.17 for a proof) that the above alternate projection algorithm is precisely the Sinkhorn-Knopp algorithm. 
Therefore, at least in this sense, the right metric to measure the distance to optimality is not the $\ell_1$ or the $\ell_2$ error as described in the previous section, but the rather the KL-divergence between the normalized vectors as described below.

Let $\piR^{(t)} := \br^{(t)}/\total$ be the $n$-dimensional probability vector whose $i$th entry is $\br^{(t)}_i/\total$; similarly define the $m$-dimensional vector $\piC^{(t)}$. 
Let $\boneR$ denote the $n$-dimensional probability vector with the $i$th entry being $\br_i/\total$; similarly define $\boneC$.  Recall that the KL-divergence between two probability distributions $p,q$ is defined as $\bD_{KL}(p||q) := \sum_{i=1}^n p_i\ln(q_i/p_i)$. The following theorem gives the convergence time for the KL-divergence.
\begin{theorem}\label{thm:main}
	If the matrix $A\in \mathbb{R}^{n\times m}_{\geq 0}$ is $(\br,\bc)$-scalable, then for any $\delta> 0$ there is a 
	$t\leq T= \ceil{\left(\frac{\ln(1+2\Delta\rho/\nu)}{\delta}\right)}$ with either $\bD_{KL}(\boneR||\piR^{(t)}) \leq \delta$ or $\bD_{KL}(\boneC||\piC^{(t)}) \leq \delta$. 
	Recall,  $\rho = \max(\max_i \br_i, \max_j \bc_j)$, $\nu = \frac{\min_{i,j:A_{ij}>0} A_{ij}}{\max_{i,j} A_{ij}}$, and $\Delta = \max_j |\{i: A_{ij} > 0\}|$ is the maximum number of non-zeros in any column of $A$.
\end{theorem}


\begin{proof}
	Let $Z := RAS$ be a matrix with row-sums $\br$ and column-sums $\bc$ for diagonal matrices $R,S$.
	Recall $A^{0}$ is the matrix obtained by column-scaling $A$. Note that the minimum non-zero entry of $A^{0}$ is $\geq \nu/\Delta$.
	\begin{lemma}\label{lem:1}
		$\bD(Z,A^{0}) \leq \ln(1+ 2\Delta\rho/\nu)$ and $\bD(Z,A^{t}) \geq 0$ for all $t$.
	\end{lemma}
	\begin{proof}
		By definition,
		\[
		\bD(Z,A^{(t)}) = \frac{1}{\total}\sum_{j=1}^m \sum_{i=1}^n Z_{ij}\ln\left(\frac{Z_{ij}}{A^{(t)}_{ij}}\right) = \frac{1}{\total}\sum_{j=1}^m \bc_j \sum_{i=1}^n \frac{Z_{ij}}{\bc_j}\ln\left(\frac{Z_{ij}}{A^{(t)}_{ij}}\right)
		\]
		For a fixed $j$, the vectors $\left(\frac{Z_{1j}}{\bc_j},\frac{Z_{2j}}{\bc_j},\ldots,\frac{Z_{nj}}{\bc_j}\right)$ and $\left(\frac{A^{(t)}_{1j}}{\bc_j},\frac{A^{(t)}_{2j}}{\bc_j},\ldots,\frac{A^{(t)}_{nj}}{\bc_j}\right)$ are probability vectors,
		and therefore the above is a sum of $\bc_j$-weighted KL-divergences which is always non-negative. 
		For the upper bound, one can use the fact (Inequality 27, \cite{jeju2015}) that for any two distributions $p$ and $q$, 
		$D(p||q) \leq \ln(1 + \frac{||p-q||^2_2}{q_\textrm{min}}) \leq \ln(1 + \frac{2}{q_\textrm{min}}) 	$ where $q_\textrm{min}$ is the smallest non-zero entry of $q$.
		For our purpose, we note that the minimum non-zero probability of the $A^{(0)}_j$ distribution being $\geq \nu/\Delta\rho$. Therefore, the second summand is at most $\ln(1+2\Delta\rho/\nu)$ giving us
		$D(Z,A^{(0)}) \leq \frac{1}{\total}\sum_{j=1}^m \bc_j \cdot \ln(1+2\Delta\rho/\nu)  = \ln(1+2\Delta\rho/\nu)$.
	\end{proof}
	
	\begin{lemma}\label{lem:2}
		\[
		\bD(Z,A^{(t)}) - \bD(Z,B^{(t)}) = \bD_{KL}(\boneR||\piR^{(t)}) ~~~ \textrm{and} ~~~ \bD(Z,B^{(t)}) - \bD(Z,A^{(t+1)}) = \bD_{KL}(\boneC||\piC^{(t)})
		\]
	\end{lemma}
	\begin{proof}
		The LHS of the first equality is simply 
		\begin{align}
		\frac{1}{\total} \sum_{j=1}^m \sum_{i=1}^n Z_{ij} \ln\left(\frac{B^{(t)}_{ij}}{A^{(t)}_{ij}}\right) && = && \frac{1}{\total} \sum_{j=1}^m \sum_{i=1}^n Z_{ij} \ln \left(\frac{\br_i}{\br^{(t)}_i}\right)\notag \\
		&& = && \frac{1}{\total} \sum_{i=1}^n \ln \left(\frac{\br_i}{\br^{(t)}_i}\right) \sum_{j=1}^m Z_{ij} \notag\\
		&& = && \sum_{i=1}^n \left(\frac{\br_i}{\total}\right)\cdot \ln \left(\frac{\br_i/\total}{\br^{(t)}_i/\total}\right)\notag
		\end{align}
		since $\sum_{j=1}^m Z_{ij} = \br_i$. The last summand is precisely $\bD_{KL}(\boneR||\piR^{(t)})$. The other equation follows analogously.
	\end{proof}
	\noindent
	The above two lemmas easily imply the theorem. 
	If for all $0\leq t\leq T$, both $\bD_{KL}(\boneR||\piR^{(t)})>\delta$ and $\bD_{KL}(\boneC||\piC^{(t)})>\delta$, then substituting in Lemma~\ref{lem:2} and summing we get
	$
	\bD(Z,A^{(0)}) - \bD(Z,A^{(T+1)}) > T\delta > \ln(1+2\Delta\rho/\nu)$
	contradicting Lemma~\ref{lem:1}. 
\end{proof}

Theorem~\ref{thm:1} follows from Theorem~\ref{thm:main} using connections between the KL-divergence and the $\ell_1$ and $\ell_2$ norms. One is the following famous Pinsker's inequality which allows us to easily prove part 1 of Theorem~\ref{thm:1}.
Given any two probability distributions $p,q$,
\begin{equation}\label{eq:pinsker}
\bD_{KL}(p||q) \geq \frac{1}{2}\cdot ||p-q||_1^2 \tag{\bf Pinsker}
\end{equation}
\begin{proof}[{\bf Proof of Theorem~\ref{thm:1}, Part 1}]
	
	Apply \eqref{eq:pinsker} on the vectors $\piR$ and $\piR^{(t)}$ to get
	\[
	\bD_{KL}(\piR||\piR^{(t)}) \geq \frac{1}{2h^2}||\br^{(t)} - \br||^2_1
	\]
	Set $\delta := \frac{\eps^2}{2h^2}$ and apply Theorem~\ref{thm:main}. In $O\left(\frac{h^2\ln(\Delta\rho/\nu)}{\eps^2}\right)$ time we would get a matrix with 
	$\delta > \bD_{KL}(\piR||\piR^{(t)})$ which from the above inequality would imply $||\br^{(t)} - \br||_1 \leq \eps$.
\end{proof}

To prove Part 2, we need a way to relate the $\ell_2$ norm and the KL-divergence. 
In order to do so, we prove a different lower bound which implies Pinsker's inequality (with a worse constant), but is significantly stronger in certain regimes. This may be of independent interest in other domains. Below we state the version
which we need for the proof of Theorem~\ref{thm:1}, part 2. This is an instantiation of the general inequality Lemma~\ref{lem:gen-pinsker} whcih we prove in Section~\ref{sec:kl}.
\begin{lemma}\label{lem:kl}
	Given any pair of probability distributions $p,q$ over a finite domain, define $\cA := \{i: q_i > 2p_i\}$ and $\cB := \{i: q_i\leq 2p_i \}$.
	Then, 
	\begin{equation}\label{eq:pinskerp}
	\bD_{KL}(p||q) \geq (1-\ln 2)\cdot \left( \sum_{i\in \cA} |q_i - p_i| + \sum_{i\in \cB}\frac{(q_i-p_i)^2}{p_i} \right)\tag{\bf KL vs $\ell_1/\ell_2$}
	\end{equation}
\end{lemma}
\begin{proof}[{\bf Proof of Theorem~\ref{thm:1}, Part 2}]
	We apply Lemma~\ref{lem:kl} on the vectors $\piR$ and $\piR^{(t)}$.
	Lemma~\ref{lem:kl} gives us
	
	\begin{align}
	\bD_{KL}(\piR||\piR^{(t)}) && \geq && C\cdot \left(\frac{1}{h}\sum_{i\in A} |\br^{(t)}_i - \br_i| ~+~ \frac{1}{h}\sum_{i\in B} \frac{(\br^{(t)}_i - \br_i)^2}{\br_i}  \right) \notag\\ && \geq && \frac{C}{h} \left( \sum_{i\in A} |\br^{(t)}_i - \br_i| ~+~ \frac{1}{\rho} \sum_{i\in B}(\br^{(t)}_i - \br_i)^2 \right) \notag
	\end{align}
	where $C = 1-\ln 2$.
	If the second summand in the parenthesis of the RHS is $\geq \frac{1}{2}||\br^{(t)} - \br||^2_2$, then we get 
	$\bD_{KL}(\piR||\piR^{(t)}) \geq \frac{C}{2\rho h} ||\br^{(t)} - \br||^2_2$.
	Otherwise, we have
	$\bD_{KL}(\piR||\piR^{(t)}) \geq \frac{C}{\sqrt{2}h}||\br^{(t)} - \br||_2$, where we used the weak fact that the sum of some positive numbers is at least the square-root of the sum of their squares. In any case, we get the following
	\begin{equation}\label{eq:min}
	\bD_{KL}(\piR||\piR^{(t)}) \geq \min\left(\frac{C}{2\rho h} ||\br^{(t)} - \br||^2_2,\frac{C}{\sqrt{2}h}||\br^{(t)} - \br||_2 \right)
	\end{equation}
	To complete the proof of part 2 of Theorem~\ref{thm:1}, set $\delta := \frac{C}{2\rho h \left(\frac{1}{\eps} + \frac{1}{\eps^2}\right)}$ 
	and apply Theorem~\ref{thm:main}. In $O\left(\rho h \ln\left(\Delta\rho/\nu\right) \cdot \left(\frac{1}{\eps} + \frac{1}{\eps^2}\right)\right)$ time we would get a matrix with 
	$\delta \geq \bD_{KL}(\piR||\piR^{(t)})$. 
	If the minimum of the RHS of \eqref{eq:min} is the first term, then we get $||\br^{(t)}-\br||^2_2 \leq \eps^2$ implying the $\ell_2$-error is $\leq \eps$. If the minimum is the second term, then we get $||\br^{(t)}-\br||_2 \leq \frac{\eps}{\sqrt{2}\rho} < \eps$ since $\rho\geq 1$.
\end{proof}

\section{New Lower Bound on the KL-Divergence}\label{sec:kl}

We now establish a new lower bound on KL-divergence which yields \eqref{eq:pinskerp} as a corollary.
\begin{lemma}\label{lem:gen-pinsker}
	Let $p$ and $q$ be two distributions over a finite $n$-element universe. For any fixed $\theta > 0$, define the sets
	$\cA_\theta := \{i\in [n]: q_i\geq (1+\theta)p_i\}$ and $\cB_\theta = [n]\setminus \cA_\theta = \{i\in [n]: q_i \leq (1+\theta)p_i\}$. 
	Then we have the following inequality
	\begin{equation}\label{eq:gen-pinsker}
	\bD_{KL}(p||q)  \geq \left(1 - \frac{\ln(1+\theta)}{\theta}\right)\cdot \left(\sum_{i\in \cA_\theta} |q_i - p_i| + \frac{1}{\theta} \sum_{i\in \cB_\theta} p_i\left(\frac{q_i - p_i}{p_i}\right)^2 \right) 
	\end{equation}
\end{lemma}
\noindent
When $\theta = 1$, we get \eqref{eq:pinskerp}.


%
\begin{proof}[{\bf Proof of Lemma~\ref{lem:gen-pinsker}:}]
	We need the following fact which follows from calculus; we provide a proof later for completeness.
	\begin{lemma}\label{fact:1}\label{fact:2}
		Given any $\theta > 0$, define $a_\theta := \frac{\ln(1+\theta)}{\theta}$ and $b_\theta := \frac{1}{\theta}\left(1 - \frac{\ln(1+\theta)}{\theta}\right)$. Then,
		\vspace{-5pt}
		\begin{itemize}
			\item For $t\geq \theta$, $(1+t) \leq e^{a_\theta t}$
			\item For $t \leq \theta$, $(1+t) \leq e^{t - b_\theta t^2}$
		\end{itemize}
	\end{lemma}
	
	\noindent
	Define $\eta_i := \frac{q_i - p_i}{p_i}$. Note that $\cA_\theta = \{i: \eta_i > \theta\}$ and $\cB_\theta$ is the rest.
	We can write the KL-divergence as follows
	\[
	\bD_{KL}(p||q) := \sum_{i=1}^n p_i\ln(p_i/q_i) = - \sum_{i=1}^n p_i \ln(1+\eta_i)
	\]
	For $i\in \cA_\theta$, since $\eta_i > \theta$, we upper bound $(1+\eta_i) \leq  e^{a_\theta\eta_i}$ using Lemma~\ref{fact:1}. 
	For $i\in \cB_\theta$, that is $\eta_i\leq \theta$, we upper bound $(1+\eta_i) \leq e^{\eta_i - b_\theta\eta^2_i}$ using Lemma~\ref{fact:2}.
	Lastly, we note $\sum_i p_i\eta_i = 0$ since $p,q$ both sum to $1$, implying $\sum_{i\in \cB_\theta}p_i\eta_i = -\sum_{i\in \cA_\theta}p_i\eta_i$. Putting all this in the definition above we get
	\[
	\bD_{KL}(p||q) \geq  - a_\theta\cdot \sum_{i\in \cA_\theta} p_i\eta_i - \sum_{i\in \cB_\theta} p_i\eta_i + b_\theta\sum_{i\in \cB_\theta}p_i\eta^2_i 
	= (1-a_\theta) \sum_{i\in \cA_\theta}p_i\eta_i + b_\theta \sum_{i\in \cB_\theta} p_i\eta^2_i 
	\]
	The proof of inequality~\eqref{eq:gen-pinsker} follows by noting that $b_\theta = \frac{1-a_\theta}{\theta}$.
\end{proof}
\noindent

\begin{proof}[Proof of Lemma~\ref{fact:1}]
	The proof of both facts follow by proving non-negativity of the relevant function in the relevant interval.
	Recall $a_\theta = \ln(1+\theta)/\theta$ and $b_\theta = \frac{1}{\theta}(1-a_\theta)$.
	We start with the following three  inequalities
	about the log-function.
	\begin{equation}\label{eq:easier}
	\textrm{For all $z>0$,}~~~~
	z + z^2/2 > (1+z)\ln(1+z) > z ~~~~ \textrm{ and} ~~~~~
	\ln(1+z) > z - z^2/2 
	\end{equation}
	The third inequality in \eqref{eq:easier} implies $a_\theta > 1 - \theta/2$ and thus, $b_\theta < 1/2$.
	The first inequality in \eqref{eq:easier} implies $a_\theta < \frac{1+\frac{\theta}{2}}{1+\theta}$ which in turn implies $b_\theta > 1/2(1+\theta)$.
	For brevity, henceforth let us lose the subscript on $a_\theta$ and $b_\theta$. 
	
	Consider the function $f(t) = e^{at} - (1+t)$. Note that $f'(t) = ae^{at} - 1$ which is increasing in $t$ since $a>0$. 
	So, for any $t\geq \theta$, we have $f'(t) \geq ae^{a\theta}-1 = \frac{(1+\theta)\ln(1+\theta)}{\theta}-1 \geq 0$, by the second inequality in \eqref{eq:easier}.
	Therefore,
	$f$ is increasing when $t\geq \theta$. The first part of Fact~\ref{fact:1} follows since $f(\theta) = 0$ by definition of $a$. \medskip
	
	Consider the function $g(t) = e^{t(1-bt)} - (1+t)$. Note that $g(0) = g(\theta) = 0$. 	
	We break the argument in two parts: we argue that $g(t)$ is strictly positive for all $t\leq 0$, and that $g(t)$ is strictly positive for $t\in (0,\theta)$. This will prove the second part of Fact~\ref{fact:1}. 
	
	The first derivative is $g'(t) = (1-2bt)e^{t(1-bt)} - 1$ and the second derivative is $g''(t) = e^{t(1-bt)}\left((1-2bt)^2 - 2b\right)$. Since $b < 1/2$, we have $2b < 1$, and thus for $t\leq 0$, $g''(t) > 0$. Therefore, $g'$ is strictly increasing for $t\leq 0$. However, $g'(0) = 0$, and so $g'(t) < 0$ for all $t< 0$. This implies $g$ is strictly decreasing in the interval $t< 0$.
	Noting $g(0) = 0$, we get $g(t) > 0$ for all $t< 0$. This completes the first part of the argument.
	
	For the second part, we first note that $g'(\theta) < 0$ since $b > \frac{1}{2(1+\theta)}$.
	That is, $g$ is strictly decreasing at $\theta$. On the other hand $g$ is increasing at $\theta$. To see this, looking at $g'$ is not enough since $g'(0) = 0$. However, $g''(0) > 0$ since $b<1/2$. This means that $0$ is a strict (local) minimum for $g$ implying $g$ is increasing at $0$. In sum, $g$ vanishes at $0$ and $\theta$, and is increasing at $0$ and decreasing at $\theta$.
	This means that if $g$ does vanish at some $r\in (0,\theta)$, then it must vanish once again in $[r,\theta)$ for the it to be decreasing at $\theta$. In particular, $g'$ must vanish three times in $(0,\theta)$ and thus four times in $[0,\theta)$ since $g'(0) = 0$. This in turn implies $g''$ vanishes three times in $[0,\theta)$ which is a contradiction since $g''$ is a quadratic in $t$ multiplied by a positive term.\smallskip
	
	We end by proving \eqref{eq:easier}. This also follows the same general methodology. Define $p(z) := (1+z)\ln(1+z) - z$ and $q(z) := p(z) - z^2/2$. Differentiating, 
	we get $p'(z) = \ln(1+z) > 0$ for all $z>0$, and $q'(z) = \ln(1+z)-z < 0$ for all $z>0$. Thus, $p$ is increasing, and $q$ is decreasing, in $(0,\infty)$. The first two inequalities of \eqref{eq:easier} follow since $p(0) = q(0) = 0$. 
	To see the third inequality, define $r(z) = \ln(1+z) - z + z^2/2$ and observe $r'(z) = \frac{1}{1+z} - 1 + z = \frac{z^2}{1+z}$
	which is $>0$ if $z>0$. Thus $r$ is strictly increasing, and the third inequality of \eqref{eq:easier} follows since $r(0)=0$.
	
\end{proof}

\subsection{Comparison with other well-known inequalities}
We connect \eqref{eq:pinskerp} with two well known lower bounds on the KL-Divergence. 
First we compare with Pinsker's inequality~\eqref{eq:pinsker}.
To see that \eqref{eq:pinskerp} generalizes \eqref{eq:pinsker} with a weaker constant, note that 
\[
||p-q||^2_1 = \left(\sum_{i\in \cA}|q_i-p_i| + \sum_{i\in \cB}|q_i-p_i|\right)^2 \leq 2\left(\sum_{i\in \cA}|q_i-p_i|\right)^2 + 2\left(\sum_{i\in \cB}p_i \frac{|q_i-p_i|}{p_i}\right)^2
\]
The first parenthetical term above, since it is $\leq 1$, is at most the first summation in the parenthesis of \eqref{eq:pinskerp}.
The second parenthetical term above, by Cauchy-Schwarz, is at most the second summation in the parenthesis of \eqref{eq:pinskerp}. 
Thus \eqref{eq:pinskerp} implies 
\[
\bD_{KL}(p||q) \geq \frac{(1-\ln 2)}{2}||p-q||^2_1
\] 
On the other hand, the RHS of \eqref{eq:pinskerp} can be much larger than that of \eqref{eq:pinsker}.
For instance, suppose $p_i=1/n$ for all $i$, $q_1 = 1/n + 1/\sqrt{n}$, and for $i\neq 1$, $q_i = 1/n - \frac{1}{(n-1)\sqrt{n}}$.
The RHS of \eqref{eq:pinsker} is $\Theta(1/n)$ while that of \eqref{eq:pinskerp} is $\Theta(1/\sqrt{n})$ which is the correct order of magnitude for $\bD_{KL}(p||q)$. \medskip

The KL-divergence between two distributions is also at least the {\em Hellinger distance} between them. Before proceeding, let us define this distance.
\[
\textrm{Given two distributions $p,q$ over $[n]$,} ~~~ \bD_\mathsf{Hellinger}(p,q) := \left(\sum_{i=1}^n \left(\sqrt{p_i} - \sqrt{q_i}\right)^2\right)^{1/2}
\]
The following inequality is known (see Reiss~\cite{reiss} p 99, Pollard~\cite{pollard-book} Chap 3.3, or the webpage~\cite{Harsha} for a proof).\
\begin{equation}
\label{eq:KL-vs-Hellinger}
\textrm{For any two distributions $p,q$,}~~~ \bD_{KL}(p||q) \geq \bD^2_\mathsf{Hellinger}(p,q) \tag{\bf KL-vs-Hellinger}
\end{equation}
It seems natural to compare the RHS of \eqref{eq:pinskerp} and  \eqref{eq:KL-vs-Hellinger} (we thank Daniel Dadush for bringing this to our attention).
As the subsequent calculation shows, the RHS of \eqref{eq:pinskerp} is in fact $\Theta(\bD^2_\mathsf{Hellinger}(p,q))$. In particular, this implies one can obtain (by reverse engineering the argument below) part 2 of Theorem 2 via the application of \eqref{eq:KL-vs-Hellinger} as well. \\

\noindent
For the set $\cA = \{i: q_i > 2p_i\}$, we know $\sqrt{q_i} + \sqrt{p_i} = \Theta(\sqrt{q_i} - \sqrt{p_i})$.
Therefore, 
\[
\sum_{i\in \cA} (q_i - p_i) = \sum_{i\in A} \left(\sqrt{q_i} + \sqrt{p_i}\right)\left(\sqrt{q_i} - \sqrt{p_i}\right) = \Theta\left(\sum_{i\in \cA} \left(\sqrt{q_i} - \sqrt{p_i}\right)^2\right)
\]
\noindent
For any $i\in \cB = \{i: q_i \le 2p_i\}$, let $q_i = (1+\eta_i)p_i$ where $-1 \le \eta_i \leq 1$. 
Via a Taylor series expansion it is not hard to check
$\left(1 + \frac{\eta_i}{2} - \sqrt{1+\eta_i}\right) = \Theta(\eta_i^2)$ in this range of $\eta_i$.
Observing that
\[
p_i\left(\frac{q_i - p_i}{p_i}\right)^2 = \eta_i^2p_i ~~~~~~ \textrm{ and } ~~~~~~ 
\left(\sqrt{p_i} - \sqrt{q_i}\right)^2 = 2p_i\left(1 + \frac{\eta_i}{2} - \sqrt{1+\eta_i} \right)
\]
we get that the RHS of \eqref{eq:pinskerp} is $\Theta(\bD^2_{\mathsf{Hellinger}}(p,q))$.

\subsection*{Acknowledgements}

We thank Daniel Dadush for asking the connection between our inequality and Hellinger distance, and Jonathan Weed for letting us know of~\cite{AWR-NIPS17}.


\bibliographystyle{abbrv}
\bibliography{refs}


\end{document}